\documentclass[conference]{IEEEtran}
\usepackage{moreverb}
\usepackage{epsfig}
\usepackage{amsmath,amssymb,amsthm,mathrsfs,amsfonts,dsfont}
\usepackage{adjustbox,lipsum}
\usepackage{algorithm,algorithmic}
\usepackage{amsfonts}
\usepackage{epsfig}
\usepackage{amssymb}
\usepackage{amsmath}
\usepackage{amsthm}
\usepackage{multirow}
\usepackage{setspace}
\usepackage{rotating}
\usepackage{graphicx}
\usepackage{tabularx}
\usepackage{array}
\usepackage{anyfontsize}
\usepackage{color,soul}
\usepackage{graphicx,dblfloatfix}
\usepackage{epstopdf}
\usepackage{blindtext}
\usepackage{amsmath}
\usepackage{amsthm,amssymb,amsmath,bm}
\usepackage{subfigure}
\usepackage{subcaption}
\usepackage{amsfonts}
\usepackage{epsfig}
\usepackage{amssymb}
\usepackage{amsmath}
\usepackage{cite}
\usepackage{graphicx}
\usepackage{fancyhdr}
\usepackage[font={small}]{caption}
\usepackage{tabularx}
\usepackage{tcolorbox}
\usepackage{cite}
\usepackage{setspace}
\usepackage[a4paper,
            bindingoffset=0.2in,
            left=0.5in,
            right=0.5in,
            top=1.15in,
            bottom=1.1in,
            footskip=.2in
            ]{geometry}

\newtheorem{theorem}{Theorem}
\newtheorem{lemm}{Lemma}

\newtheorem{corollary}{Corollary}

\newtheorem{rem}{Remark}
\newtheorem{Pro}{Proposition}

\newcommand{\E}[1]{\mathbb{E}[#1]}
\newcommand{\nn}{\nonumber\\}
\newcommand{\Dbar}{\bar{D}}

\allowdisplaybreaks
\begin{document}

\title{
Multi-Source M/G/1/1 Queues with Probabilistic Preemption
}
 \author{
\IEEEauthorblockN{Mohammad~Moltafet\textsuperscript{*}, Hamid R. Sadjadpour\textsuperscript{*}, Zouheir~Rezki\textsuperscript{*}, Marian~Codreanu\textsuperscript{†}, and Roy~D.~Yates\textsuperscript{††}
 \\
\textsuperscript{*}Department of ECE, University of California Santa Cruz, USA 
(\{mmoltafe, hamid, zrezki\}@ucsc.edu)
\\
\textsuperscript{†}Department of Science and Technology, Link\"{o}ping University, Sweden (marian.codreanu@liu.se)
\\
\textsuperscript{††}Department of ECE, Rutgers University, USA (ryates@winlab.rutgers.edu)
}
}

\maketitle
\begin{abstract}
We consider a multi-source status update system consisting of multiple independent sources, a single server, and a single sink. Each source generates packets according to a Poisson process, and packets are served according to a general service time distribution. The system has a capacity of one packet, i.e., no waiting buffer, and is modeled as a multi-source M/G/1/1 queueing system. We introduce a probabilistically preemptive packet management policy, under which an existing packet from the same source in the system is replaced by an arriving packet with a fixed probability. We derive the moment generating functions (MGFs) of the age of information (AoI) and peak AoI (PAoI) for each source under this policy. Numerical results demonstrate the effectiveness of the proposed packet management policy.

\end{abstract}	
\section{Introduction}\label{Introduction}
The authors of \cite{6195689} introduced the age of information (AoI) as a destination-centric metric to quantify information freshness in status update systems. The AoI represents the difference between the current time and the time stamp of the most recently received sample from the monitored process. The peak AoI (PAoI) was later proposed in \cite{6875100} as an alternative metric for evaluating information freshness, measuring the value of the AoI immediately before packet delivery.

This paper considers a multi-source status update system with a single server and no waiting buffer. The packets of each source are generated according to a Poisson process, and each packet is served according to a general service time distribution. We derive the moment generating functions (MGFs) of the AoI and PAoI under a probabilistically preemptive policy. Under this policy, when a packet arrives, an existing packet from the \textit{same source} in the system is preempted by the arriving packet with a fixed probability.

Using the MGFs of the AoI, the average AoI of a two-source queueing system is analyzed in the numerical results section. The results indicate that, depending on system parameters such as packet arrival rates and service time distribution parameters, the proposed probabilistically preemptive packet management policy can outperform the self-preemptive \cite{9869867}, globally preemptive \cite{8406928}, and non-preemptive \cite{9500775} policies. Moreover, the results demonstrate that when a preemption mechanism is available, the system must learn when to apply it to achieve improved performance.

\subsection{Related Work}
The work in \cite{6195689} derived the average AoI for M/M/1, D/M/1, and M/D/1 first-come first-served (FCFS) queueing models.  
The authors of \cite{6284003} investigated the AoI in a multi-source setup where they studied the average AoI in a multi-source M/M/1 FCFS queueing model.
As shown in the initial works \cite{6310931,7415972}, applying an appropriate packet management policy in status update systems -- in the waiting queue or/and server -- has a substantial potential to improve information freshness. 
The performance of various packet management policies in queueing systems with exponentially distributed service times and Poisson arrivals has been extensively studied \cite{8469047,8437591,8406966,8437907,9013935,9048914,9252168,9162681,Moltafet2020mgf,9611498,9705518}.

Beyond the setting of exponentially distributed service times and Poisson arrivals, AoI has been studied under a variety of arrival processes and service time distributions in both single- and multi-source systems.
 The distributions of the AoI and PAoI for the \textit{single-source} PH/PH/1/1 and M/PH/1/2 queueing models were derived in \cite{9119460}.
The average AoI of a single-source D/G/1 FCFS queueing model was derived in \cite{8406909}. A closed-form expression for the average AoI of a single-source M/G/1/1 preemptive queueing model with hybrid automatic repeat request was provided in \cite{8006504}. The AoI and PAoI distributions for single-source M/G/1/1 and G/M/1/1 queueing models were obtained in \cite{8006592}. A general formula for the AoI distribution in single-source, single-server queueing systems was derived in \cite{8820073}.
The average AoI and PAoI of a single-source status-update system with Poisson arrivals and gamma-distributed service times under a last-come first-served (LCFS) packet management policy were analyzed in \cite{7541764}.
The average AoI of a single-source G/G/1/1 queueing model was studied in \cite{9048933}. The AoI distribution of a generate-at-will single-source dual-server system was derived in \cite{10899900}, where the servers were assumed to have exponentially distributed service times and the sampling/transmission process was frozen for an Erlang-distributed duration after each transmission. The MGFs of the AoI and PAoI for a single-source M/G/1/1 model under a probabilistically preemptive policy were derived in \cite{moltafetisit2025}.

The average AoI and PAoI in a \textit{multi-source} M/G/1 FCFS queueing model were studied in \cite{9099557,inoue2024exact}. The average AoI of a queueing system with two classes of Poisson arrivals with different priorities under a general service time distribution was analyzed in \cite{8886357}. The average AoI and PAoI of a multi-source M/G/1/1 queueing model under the globally preemptive packet management policy were derived in \cite{8406928}. Under this policy, a newly arriving packet preempts any packet in service, regardless of the source index.
The distributions of the AoI and PAoI for a generate-at-will multi-source system with a general phase-type service time distribution were derived in \cite{10139823}. The average AoI and PAoI of a multi-source M/G/1/1 queueing model under the non-preemptive policy were obtained in \cite{9500775}. Under this policy, any arriving packet is discarded when the server is busy, regardless of its source index.  
 The authors of \cite{9519697} considered a multi-source system with Poisson arrivals, where the server serves packets according to a phase-type distribution, i.e., a multi-source M/PH/1/1 queueing system. Using the theory of Markov fluid queues, they proposed a method to numerically obtain the distributions of the AoI and PAoI under a probabilistically preemptive policy. Under this policy, a newly arriving packet from source $c$ can preempt a packet from source $c'$ in service with a probability that depends on $c$ and $c'$. The MGFs of the AoI and PAoI of a multi-source M/G/1/1 queueing model under the self-preemptive policy were derived in \cite{9869867}, where a newly arriving packet preempts a packet in service only if they have the same source index. In addition, the authors of \cite{9869867} obtained the MGFs of the AoI and PAoI for the models studied in \cite{9500775} and \cite{9869867}.
The Laplace-Stieltjes transform of the AoI for a two-source system with Poisson arrivals and a generally distributed service time was derived in \cite{10038591}. The authors assumed that each source has its own buffer and studied three versions of the self-preemptive policy.

\subsection{Organization}
In Section~\ref{System Model and Summary of Results}, we present the system model and summarize the main results. The derivation of the MGFs of the AoI and PAoI is provided in Section~\ref{Calculation of the MGF of the (peak) AoI}. Numerical results are presented in Section~\ref{Numerical Results}, and the paper is concluded in Section~\ref{Conclusions}.

\section{System Model and Main Results}\label{System Model and Summary of Results}

We consider a status update system consisting of $C$ independent sources, denoted by the set $\mathcal{C} = \{1, \dots, C\}$, a single server, and a single sink, as illustrated in Fig.~\ref{Model}. Each source generates status update packets about a random process and sends them to the sink. Each packet contains the measured value of the monitored process and a time stamp representing the generation time of the sample. We assume that packets from source $c \in \mathcal{C}$ are generated according to a Poisson process with rate $\lambda_c$. Since the sources generate packets independently, the overall packet generation in the system follows a Poisson process with rate $\lambda = \sum_{c' \in \mathcal{C}} \lambda_{c'}$. 

The server serves packets according to a general service time distribution. More specifically, we assume that packets have i.i.d. service times with probability density function (PDF) $f_U(\cdot)$ and moment generating function (MGF) $M_U(s) = \mathbb{E}[e^{sU}]$. Finally, the system has a capacity of one (i.e., no waiting buffer), making it a multi-source M/G/1/1 queueing system. Next, we describe the packet management policy.

\begin{figure}
\centering
\includegraphics[width=.6\linewidth,trim = 0mm 0mm 0mm 0mm,clip]{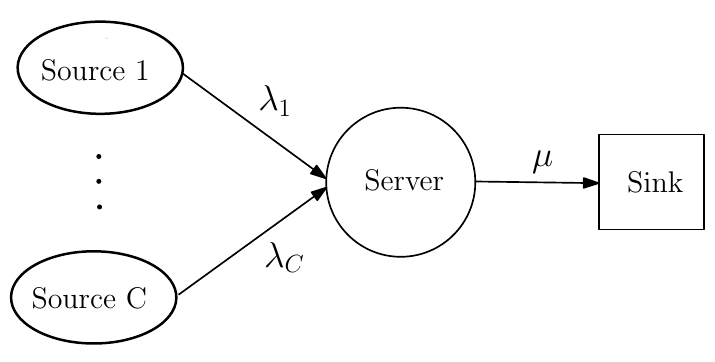}
\caption{The considered status update system.}  
\label{Model}
\vspace{-5mm}		
\end{figure}

\textbf{Probabilistically Preemptive Policy:} Under this policy, when the server is idle and a packet arrives, the arriving packet immediately enters service. When the server is busy and a new packet arrives, an existing packet from the \textit{same source} in the system is preempted by the arriving packet with probability $\theta$. If the arriving packet finds a packet from a different source in service, it is discarded.

\subsection{AoI Definition}

For each source, the AoI at the sink is defined as the time elapsed since the most recently received packet was generated. Formally, let $t_{c,i}$ denote the generation time of the $i$th successfully delivered status update packet from source $c$, and let $t'_{c,i}$ denote the time at which this packet arrives at the sink.



At time instant $\tau$, the index of the most recently received packet from source $c$ is given by
$
{N_c(\tau) = \max\{i' \mid t'_{c,i'} \le \tau\},}
$
and the generation time of this packet is
$
\xi_c(\tau) = t_{c, N_c(\tau)}.
$
The AoI of source $c$ at the sink is then defined as the stochastic process
$
\delta_c(t) = t - \xi_c(t).
$

Let the random variable
$Y_{c,i} = t'_{c,i} - t'_{c,i-1}$
denote the $i$th interdeparture time of source $c$, i.e., the time elapsed between the departures of the $(i-1)$th and $i$th successfully delivered packets from source $c$. Hereafter, we refer to the $i$th delivered packet from source $c$ simply as ``packet $c,i$''. 
Moreover, let the random variable
$T_{c,i} = t'_{c,i} - t_{c,i}$
denote the system time of packet $c,i$, i.e., the total time that the $i$th delivered packet spends in the system.

One of the most commonly used metrics for evaluating the AoI is the PAoI \cite{6875100}. The PAoI of source $c$ is defined as the value of the AoI immediately before a new update packet is received. Accordingly, the PAoI associated with packet $c,i$, denoted by $A_{c,i}$, is given by
\begin{align}\label{AoI.eq}
A_{c,i} = Y_{c,i} + T_{c,i-1}.
\end{align}

We assume that the considered status update system is stationary and that the AoI process for each source is ergodic. Thus, we have $T_{c,i} =^{\mathrm{st}}T_c$, $Y_{c,i} =^{\mathrm{st}}Y_c$, and $A_{c,i} =^{\mathrm{st}}A_c$ for all $i$, where $=^{\mathrm{st}}$ denotes stochastic equivalence, i.e., the variables have identical marginal distributions. 

We present the main results of this paper in the following theorem.


\begin{theorem}\label{T_source-aware}
The MGFs of the AoI and PAoI for source $c$ under the probabilistically preemptive packet management policy, denoted by $M_{\delta_c}(s)$ and $M_{A_c}(s)$, respectively, are given by
\begin{align}
&{M_{\delta_c}(s)}=\dfrac{M_{T_c}(s)(M_{Y_c}(s)-1)}{s\bar{Y}_c},\\&
M_{A_c}(s)=M_{T_c}(s)M_{Y_c}(s),
\end{align}
where $M_{T_c}(s)$ is the MGF of the system time $T_c$ of source $c$, which is given by 
\begin{align}\label{STMGF}
M_{T_c}(s)=\dfrac{M_U(s-\theta\lambda_c)}{M_{U}(-\theta\lambda_c)},
\end{align}
 $M_{Y_c}(s)$ is the MGF of the interdeparture time $Y_c$ of source $c$, which is given by 
\begin{align}\label{mgfinterdeparty0}
M_{Y_c}(s)=\dfrac{a_c}{(1-a'_c)\left(1-\sum_{c'\in\mathcal{C}\setminus\{c\}}\dfrac{a_{c'}}{1-a'_{c'}}\right)},
\end{align}
where $a_c=\dfrac{\lambda_cM_U(s-\theta\lambda_c)}{\lambda-s}$, $a'_c=\dfrac{\theta\lambda_c(1-M_U(s-\theta\lambda_c))}{\theta\lambda_c-s}$, and $\bar{Y}_c$ is the mean of $Y_c$ which is derived by calculating the first derivative of the MGF of $Y_c$, evaluated at $s=0$, i.e.,
$
\bar{Y}_c=\dfrac{\mathrm{d}(M_{Y_c}(s))}{\mathrm{d}s}\Big|_{s=0}.
$
\end{theorem}
\begin{proof}
    See Section~\ref{Calculation of the MGF of the (peak) AoI}.
\end{proof}

\begin{corollary}\label{Cro_01}
    The $m$th moments of the AoI and PAoI for source $c$, denoted by $\Delta_c^{(m)}$ and  $A_c^{(m)}$, respectively, are given by 
    \begin{align}
       &\Delta_c^{(m)}=\dfrac{\sum_{i=0}^{m+1}{m+1 \choose i}\E{T_c^{i}}\E{Y_c^{m+1-i}}-\E{T_c^{m+1}}}{(m+1)\bar{Y_c}},\\&
       A_c^{(m)}=\sum_{i=0}^{m}{m\choose i}\E{T_c^{i}}\E{Y_c^{m-i}},
    \end{align}
    where the $j$th moment of $T_c$ (resp. $Y_c$)  is derived by evaluating the $j$th derivative of the $M_{T_c}(s)$ in \eqref{STMGF} (resp. $M_{Y_c}(s)$ in \eqref{mgfinterdeparty0}) at $s=0$.
\end{corollary}
\begin{proof}
    See Section~\ref{ProofCro_01}.
\end{proof}

\begin{rem}
The MGF of the (peak) AoI under the probabilistically preemptive policy, presented in Theorem~\ref{T_source-aware}, generalizes the existing results in \cite{9869867} and \cite{moltafetisit2025}. Specifically, by letting $\theta \rightarrow 0$, the MGF of the (peak) AoI reduces to that under the non-preemptive policy derived in \cite{9869867}; by letting $\theta \rightarrow 1$, it reduces to that under the self-preemptive policy derived in \cite{9869867}. Furthermore, by letting $\lambda_{c'} \rightarrow 0$ for all $c' \in \mathcal{C} \setminus \{c\}$, the MGF of the (peak) AoI reduces to that of a single-source system with arrival rate $\lambda_c$ under the probabilistically preemptive policy derived in \cite{moltafetisit2025}.
\end{rem}

\section{Derivation of the MGF of the (Peak) AoI }\label{Calculation of the MGF of the (peak) AoI}
To prove Theorem~\ref{T_source-aware}, we first present Lemma~\ref{lemmsmgfage}, which expresses the MGFs of the AoI and PAoI for source $c$ as functions of the MGF of the system time $T_c$ and the MGF of the interdeparture time $Y_c$ for source $c$.

\begin{lemm}\label{lemmsmgfage}
The MGFs of the AoI and PAoI for source $c$ in a multi-source M/G/1/1 queueing model under the probabilistically preemptive policy can be expressed as
\begin{align}\label{MGFofagegeneral}
&M_{\delta_{c}}(s)=\dfrac{M_{A_c}(s)-M_{T_c}(s)}{s\bar{Y}_c},\\&\label{MGFpeak}
M_{A_c}(s)=M_{T_c}(s)M_{Y_c}(s).  
\end{align}
\end{lemm}
\begin{proof} 
The proof follows similar steps as in \cite[Lemma~1]{9869867}.
\end{proof}

From Lemma~\ref{lemmsmgfage}, the main challenge in calculating the MGF of the (peak) AoI reduces to deriving the MGFs of the system time of source $c$, $M_{T_c}(s)$, and the interdeparture time of source $c$, $M_{Y_c}(s)$. 

Next, we calculate $M_{T_c}(s)$ and $M_{Y_c}(s)$ in Propositions~\ref{Pro1} and \ref{Pro2}, respectively.

\begin{Pro}\label{Pro1}
The MGF of the system time $T_c$ for source $c$ is given by
 \begin{align}\label{mgfsystemtime}
M_{T_c}(s)=\dfrac{M_S(s-\theta\lambda_c)}{L_{\theta\lambda_c}}.
 \end{align}
\end{Pro}

\begin{proof}
See Section~\ref{ProofPro1}.
\end{proof}

The next step is to derive the MGF of the interdeparture time $Y_c$, 
which is carried out as follows.

\begin{Pro}\label{Pro2}
The MGF of the interdeparture time $Y_c$ for source $c$ is given by
\begin{align}\label{mgfinterde1}
M_{Y_c}(s)=\dfrac{a_c}{(1-a'_c)\left(1-\sum_{c'\in\mathcal{C}\setminus\{c\}}\dfrac{a_{c'}}{1-a'_{c'}}\right)},
\end{align}
where ${a_c\!=\!\dfrac{\lambda_cM_U(s\!-\!\theta\lambda_c)}{\lambda\!-\!s}}$, and ${a'_c\!=\!\dfrac{\theta\lambda_c(1\!-\!M_U(s\!-\!\theta\lambda_c))}{\theta\lambda_c-s}}$.
\end{Pro}
\begin{proof}
To calculate the MGF of the interdeparture time $Y_c$, i.e., $M_{Y_c}(s) = \mathbb{E}[e^{s Y_c}]$, we first characterize $Y_c$ using a semi-Markov chain. The semi-Markov chain, shown in Fig.~\ref{Semi-Chain_c}, represents the dynamics of the system occupancy states (denoted by $q$'s) and the transition probabilities (denoted by $p$'s) between different states with respect to the interdeparture time $Y_c$.

The $C+2$ states of the graph in Fig.~\ref{Semi-Chain_c}, i.e., $\{q_0, q_1, q_2, \dots, q_C, q'_0\}$, are described as follows. When a packet from source $c$ is successfully delivered to the sink, the system enters state $q_0$, waiting for a fresh packet from any source. State $q_{c'}$, $c' \in \mathcal{C}$, indicates that a packet from source $c'$ is in service. State $q'_0$ indicates that a packet from some source $c' \in \mathcal{C}_{-c}$, where $\mathcal{C}_{-c} = \mathcal{C} \setminus \{c\}$, has been successfully delivered to the sink and the server is waiting for a fresh packet from any source. From the graph, the interdeparture time $Y_c$ is obtained by characterizing the total time required to start from state $q_0$ and return to $q_0$.

Let $\bar{X}_c = \min_{c' \in \mathcal{C}_{-c}} X_{c'}$, where $X_{c'}$ is a random variable representing the interarrival time between any two consecutive packets of source $c'$. The transitions between the states of the graph in Fig.~\ref{Semi-Chain_c} are described as follows.

\begin{enumerate}

\item $q_0 \rightarrow q_{c'},~\forall c' \in \mathcal{C}$: The system is in state $q_0$ and a packet from source $c'$ arrives. This transition occurs if the interarrival time of the source $c'$ packet, $X_{c'}$, is shorter than the minimum interarrival time among all other sources, $\bar{X}_{c'}$. Therefore, the transition occurs with probability $p_{c'} = \mathrm{Pr}(X_{c'} < \bar{X}_{c'})$. We denote the sojourn time of the system in state $q_0$ before this transition by $\eta_{c'}$, which has the distribution $ \mathrm{Pr}(\eta_{c'} > t) = \mathrm{Pr}(X_{c'} > t \mid X_{c'} < \bar{X}_{c'}) $.

\item $q'_0 \rightarrow q_{c'},~\forall c' \in \mathcal{C}$: The transition probability and the distribution of the sojourn time of the system in state $q'_0$ are the same as those for the transition $q_0 \rightarrow q_{c'}$.

\item $q_c \rightarrow q_0$: The system is in state $q_c$, i.e., serving a packet from source $c$, and the packet completes service and is delivered to the sink. This transition occurs with probability $\bar{p}_c = \mathrm{Pr}(D_c)$, where $D_c$ denotes the event that a packet from source $c$ entering service is successfully delivered. We denote the sojourn time of the system in state $q_c$ before this transition by $\bar{\eta}_c$, which has the distribution $\mathrm{Pr}(\bar{\eta}_c > t) = \mathrm{Pr}(U > t \mid D_c)$.

\item $q_{c'} \rightarrow q'_0,~\forall c' \in \mathcal{C}_{-c}$: The system is in state $q_{c'}$, $c' \in \mathcal{C}_{-c}$, i.e., serving a packet from source $c'$, and the packet completes service and is delivered to the sink. This transition occurs with probability $\bar{p}_{c'} = \mathrm{Pr}(D_{c'})$. The sojourn time of the system in state $q_{c'}$ before this transition has the distribution $\mathrm{Pr}(\bar{\eta}_{c'} > t) = \mathrm{Pr}(U > t \mid D_{c'})$.

\item $q_{c'} \rightarrow q_{c'},~\forall c' \in \mathcal{C}$: The system is in state $q_{c'}$ and a fresh packet from source $c'$, possibly following several blocked and cleared packets of the same source, arrives and preempts the packet in service. Let $\bar{D}_{c'}$ denote the event that a source $c'$ packet in service is not delivered because it is preempted. This transition occurs with probability $p'_{c'} = \mathrm{Pr}(\bar{D}_{c'})$. Since $\bar{D}_{c'}$ is the complement of $D_{c'}$, we have $p'_{c'} = 1 - \mathrm{Pr}(D_{c'})$. We denote the sojourn time of the system in state $q_{c'}$ before this transition by $\eta'_{c'}$.

\end{enumerate}

\begin{figure}
\centering
\includegraphics[width=.6\linewidth,trim = 0mm 0mm 0mm 0mm,clip]{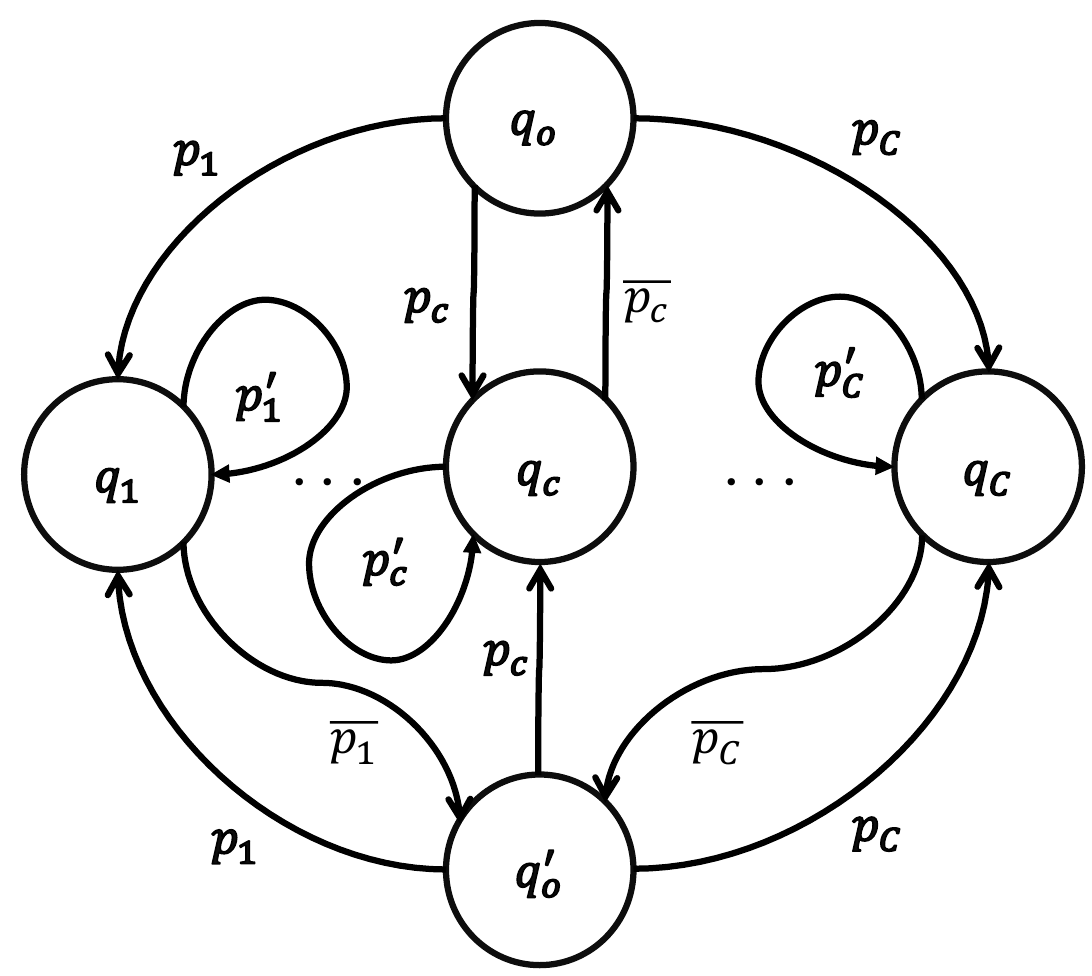}
\caption{The semi-Markov chain corresponding to the interdeparture time of two consecutive packets of source $c$.
}  
\label{Semi-Chain_c}
\vspace{-5mm}
\end{figure}

Next, we proceed to compute the transition probabilities and the corresponding sojourn time distributions for all states.

\begin{lemm}\label{prob01}
The transition probabilities $p_{c'}$, $p'_{c'}$, and $\bar p_{c'}$ for all ${c'}\in\mathcal{C}$ are characterized as follows:
\begin{align}
&p_{c'}=\dfrac{\lambda_{c'}}{\lambda},~
\bar p_{c'}=M_{U}(-\theta\lambda_{c'}),~
p'_{c'}=1-M_{U}(-\theta\lambda_{c'}).
\end{align}
\end{lemm}
\begin{proof}
See Section~\ref{Proof_LammaII}.
\end{proof}

The sojourn time distributions are calculated in the following lemma. 
\begin{lemm}\label{lemmfaap}
The PDFs of the random variables $\eta_{c'}$, $\bar\eta_{c'}$, and $\eta'_{c'}$ for all ${c'}\in\mathcal{C}$ are given as follows:
\begin{align}
&f_{\eta_{c'}}(t)=\lambda e^{-\lambda t},\\&\nonumber 
f_{\bar \eta_{c'}}(t)=\dfrac{f_U(t)e^{-\theta\lambda_{c'}t}}{M_{U}(-\theta\lambda_{c'})},\\&\nonumber
f_{\eta'_{c'}}(t)=\dfrac{\theta\lambda_{c'}e^{-\theta\lambda_{c'}t}(1-F_U(t))}{1-M_{U}(-\theta\lambda_{c'})}.
\end{align}
\end{lemm}
\begin{proof}
See Section~\ref{Proof_LammaII}.
\end{proof}

As shown in Fig.~\ref{Semi-Chain_c}, the interdeparture time between two consecutive delivered packets from source $c$ is equal to the total sojourn time experienced by the status update system from the moment it enters state $q_0$ until it returns to state $q_0$. 

The total sojourn time along this path consists of the sum of the individual sojourn times in each state for all possible paths $\{q_0, \ldots, q_0\}$. 

Thus, the random variable $Y_c$ can be characterized in terms of the random variables $\eta_{c'}$, $\bar{\eta}_{c'}$, and $\eta'_{c'}$ for all $c' \in \mathcal{C}$, i.e., the sojourn times in different states, and the number of times each occurs. Let $K_{c'}$, $\bar{K}_{c'}$, and $K'_{c'}$ denote the random variables representing the numbers of occurrences of $\eta_{c'}$, $\bar{\eta}_{c'}$, and $\eta'_{c'}$, respectively, during the interdeparture time $Y_c$. Consequently, $Y_c$ can be expressed as

\begin{equation}\label{Y_c_sojourn}
Y_c=\textstyle\sum_{c'\in\mathcal{C}}K_{c'}\eta_{c'}+\textstyle\sum_{{c'}\in\mathcal{C}}\bar K_{c'}\bar\eta_{c'}+\textstyle\sum_{{c'}\in\mathcal{C}}K'_{c'}\eta'_{c'}.
\end{equation} 

Using $Y_c$ defined in \eqref{Y_c_sojourn}, the MGF ${M}_{Y_c}(s) = \mathbb{E}[e^{sY_c}]$ can be calculated as in \eqref{mgfequ} (shown at the top of the next page),
\begin{figure*}[t!]
\begin{align}\label{mgfequ}
&{M}_{Y_c}(s)=\mathbb{E}[e^{sY_c}]=\mathbb{E}\big[\mathbb{E}[e^{sY_c}\mid (K_1,\cdots,K_C,\bar K_1,\cdots,\bar K_C,K'_1,\cdots,K'_C)=(k_1,\cdots,k_C,\bar k_1,\cdots,\bar k_C,k'_1,\cdots,k'_C)]\big]\\\nonumber
&=\sum_{k_1,\cdots,k_C,\bar k_1,\cdots,\bar k_C,k'_1,\cdots,k'_C}\!\!\!\!\!\!\mathbb{E}\big[e^{s(\sum_{{c'}\in\mathcal{C}}k_{c'}\eta_{c'}+\sum_{{c'}\in\mathcal{C}}\bar k_{c'}\bar\eta_{c'}+\sum_{{c'}\in\mathcal{C}}k'_{c'}\eta'_{c'})}\big]\\\nonumber
&\hspace{6mm}\mathrm{Pr}\big((K_1,\cdots,K_C,\bar K_1,\cdots,\bar K_C,K'_1,\cdots,K'_C)=(k_1,\cdots,k_C,\bar k_1,\cdots,\bar k_C,k'_1,\cdots,k'_C)\big)\\\nonumber
&\stackrel{(a)}{=}\sum_{k_1,\cdots,k_C,\bar k_1,\cdots,\bar k_C,k'_1,\cdots,k'_C}\prod_{c'=1}^{C}\mathbb{E}[e^{s\eta_{c'}}]^{k_{c'}}\prod_{{c'}=1}^{C}\mathbb{E}[e^{s\bar\eta_{c'}}]^{\bar k_{c'}}\prod_{{c'}=1}^{C}\mathbb{E}[e^{s\eta'_{c'}}]^{k'_{c'}}\\\nonumber
&\hspace{6mm}\prod_{{c'}=1}^{C}p_{c'}^{k_{c'}}\prod_{{c'}=1}^{C}\bar p_{c'}^{\bar k_{c'}}\prod_{{c'}=1}^{C}{p'_{c'}}^{k'_{c'}}
Q(k_1,\cdots,k_C,\bar k_1,\cdots,\bar k_C,k'_1,\cdots,k'_C),
\end{align}
\rule{\textwidth}{0.4pt}
\end{figure*}
where equality $ (a) $ follows because: i) the random variables $\eta_{c'}$, $\bar \eta_{c'}$, and $\eta'_{c'}$ for all ${c'\in\mathcal{C}}$ are independent, and ii) because of the independence of paths, $ \mathrm{Pr}\big((K_1,\cdots,K_C,\bar K_1,\cdots,\bar K_C,K'_1,\cdots,K'_C)=(k_1,\cdots,k_C,\bar k_1,\cdots,\bar k_C,k'_1,\cdots,k'_C)\big) $ 
equals the sum of the probabilities of all possible paths corresponding to the occurrence combination 
$ (k_1,\cdots,k_C,\bar k_1,\cdots,\bar k_C,k'_1,\cdots,k'_C) $, which is given as 
\begin{align}\nonumber
&\prod_{{c'}=1}^{C}p_{c'}^{k_{c'}}\prod_{{c'}=1}^{C}\bar p_{c'}^{\bar k_{c'}}\prod_{{c'}=1}^{C}{p'_{c'}}^{k'_{c'}}
\\& ~~~~~~~~~~...\times Q(k_1,\cdots,k_C,\bar k_1,\cdots,\bar k_C,k'_1,\cdots,k'_C), 
\end{align}
where $Q(k_1,\cdots,k_C,\bar k_1,\cdots,\bar k_C,k'_1,\cdots,k'_C)$
denote the number of paths  with the occurrence combination $ (k_1,\cdots,k_C,\bar k_1,\cdots,\bar k_C,k'_1,\cdots,k'_C) $. 

To derive $ {M}_{Y_c}(s)$ given in \eqref{mgfequ}, we need to calculate i) the values of $\mathbb{E}[e^{s\eta_{c'}}],~\mathbb{E}[e^{s\bar\eta_{c'}}]$, and $ \mathbb{E}[e^{s\eta_{c'}'}] $ for all ${c'}\in\mathcal{C}$, ii) the number of paths with the occurrence combination $ (k_1,\cdots,k_C,\bar k_1,\cdots,\bar k_C,k'_1,\cdots,k'_C)$, i.e.,  $Q(k_1,\cdots,k_C,\bar k_1,\cdots,\bar k_C,k'_1,\cdots,k'_C)$, and iii) the summation over the different occurrence combinations, which are carried out in the following. 

\begin{lemm}\label{rem01}
Using the PDFs presented in Lemma \ref{lemmfaap}, the expectations $\mathbb{E}[e^{s\eta_{c'}}]$, $\mathbb{E}[e^{s\bar{\eta}_{c'}}]$, and $\mathbb{E}[e^{s\eta'_{c'}}]$ are given by
\begin{align}\nonumber
&\mathbb{E}[e^{s\eta_{c'}}]=\dfrac{\lambda}{\lambda-s},\\&\nonumber
\mathbb{E}[e^{s\bar\eta_{c'}}]=\dfrac{M_U(s-\theta\lambda_{c'})}{M_{U}(-\theta\lambda_{c'})},\\&\label{MGfprob}
\mathbb{E}[e^{s\eta'_{c'}}]=\dfrac{\theta\lambda_{c'}(1-M_U(s-\theta\lambda_{c'}))}{(M_{U}(-\theta\lambda)-1)(s-\theta\lambda_{c'})}.
\end{align}
\end{lemm}

Next, we address Items ii and iii using the following lemma, which provides an effective tool from graph theory.

\begin{lemm}\label{lembro}
Consider a directed graph $G = (\mathcal{V}, \mathcal{E})$ consisting of a set $\mathcal{V}$ of $V$ nodes, a set $\mathcal{E}$ of $E$ edges, an algebraic label $e_{v' \rightarrow \bar v}$ on each edge $e \in \mathcal{E}$ from node $v'$ to node $\bar v$, and a node $u \in \mathcal{V}$ with no incoming edges. 
Let the transfer function $H(v)$ denote the weighted sum over all paths from $u$ to $v$, where the weight of each path is given by the product of its edge labels. Then, the transfer functions $H(v)$, $\forall v \in \mathcal{V}$, can be calculated by solving the following system of linear equations:
\begin{align}\label{hv}
\begin{cases}
H(u)=1\\
H(v)=\sum_{v'\in \mathcal{E}}e_{v'\rightarrow v}H(v'),& u\ne v.
\end{cases}
\end{align}
\end{lemm}
\begin{proof}
See \cite[Sect.~6.4]{6Bixio2016}.
\end{proof}

The main idea is to construct a directed graph $G = (\mathcal{V}, \mathcal{E})$ by defining its set of nodes $\mathcal{V}$, the directed edges $\mathcal{E}$ with weights $e_{v' \rightarrow \bar v}$, and the transfer functions of each node, $H(v)$, for all $v \in \mathcal{V}$, such that the expression on the right-hand side of $(a)$ in \eqref{mgfequ} is equal to the transfer function of a node $\bar{v} \in \mathcal{V}$, i.e., $H(\bar{v})$. In other words, we aim to form a directed graph $G = (\mathcal{V}, \mathcal{E})$ such that for some node $\bar{v} \in \mathcal{V}$, we have ${M}_{Y_c}(s) = H(\bar{v})$.

The design of the graph $G$ can be understood by noting its strong similarity to the structure of the directed graph in Fig.~\ref{Semi-Chain_c}, which we used to characterize the interdeparture time $Y_c$. The main difference is that all nodes in the directed graph shown in Fig.~\ref{Semi-Chain_c} have incoming edges. 

To create the node $u \in \mathcal{V}$ with no incoming edges, we remove the incoming links of node $q_0$ and introduce a virtual node $\bar{q}_0$ to represent the system state after completing the service of a source $c$ packet. 

By examining the factors representing the edge weights on the right-hand side of $(a)$ in \eqref{mgfequ}, we construct the directed graph $G$ shown in Fig.~\ref{Detour_c}. Thus, ${M}_{Y_c}(s)$ is given by the transfer function from node $q_0$ to node $\bar{q}_0$, i.e., $H(\bar{q}_0)$, so that ${M}_{Y_c}(s) = H(\bar{q}_0)$. Finally, the remaining task is to derive $H(\bar{q}_0)$ based on \eqref{hv}.
The system of linear equations in \eqref{hv} corresponding to the directed graph in Fig.~\ref{Detour_c} is formulated as
\begin{align}\nonumber
&H( q_0)=1,\\&\nonumber
H(\bar q_0)=\bar p_c\mathbb{E}[e^{s\bar \eta_c}]H(q_c),\\&\nonumber
H( q'_0)=\sum_{{c'}\in\mathcal{C}_{-c}}\bar p_{c'}\mathbb{E}[e^{s\bar \eta_{c'}}]H(q_{c'}),\\&\nonumber
H( q_{c'})=p_{c'}\mathbb{E}[e^{s\eta_{c'}}]H(q_0)+p'_{c'}\mathbb{E}[e^{s\eta'_{c'}}]H(q_{c'})\\&\label{hv0}~~~~~~~~~~~~~~~~~~~~~~~~~~~~~+ p_{c'}\mathbb{E}[e^{s\eta_{c'}}]H(q'_0),
\forall {c'}\in \mathcal{C}.
\end{align}
By solving the system of linear equations in \eqref{hv0}, $H(\bar q_0)$ is given by 
\begin{align}\label{barhq0}
&H(\bar q_0)=\\&\nonumber \dfrac{p_{c}{\bar p_{c}}\mathbb{E}[e^{s\eta_{c}}]\mathbb{E}[e^{s\bar\eta_{c}}]}{\big(1-p'_{c}\mathbb{E}[e^{s\eta'_{c}}]\big)\bigg(1-\sum_{c'\in\mathcal{C}_{-c}}\dfrac{p_{c'}{\bar p_{c'}}\mathbb{E}[e^{s\eta_{c'}}]\mathbb{E}[e^{s\bar\eta_{c'}}]}{1-p'_{c'}\mathbb{E}[e^{s\eta'_{c'}}]}\bigg)}.
\end{align}

Finally, substituting the probabilities $p_{c'}$, $p'_{c'}$, and $\bar p_{c'}$ (derived in  Lemma~\ref{prob01}) and the values of  $\mathbb{E}[e^{s\eta_{c'}}]$, $\mathbb{E}[e^{s\eta'_{c'}}]$, and $\mathbb{E}[e^{s\bar{\eta}_{c'}}]$ (derived in Lemma~\ref{rem01}) into \eqref{barhq0} results in ${M}_{Y_c}(s)$, which completes the proof of Proposition~\ref{Pro2}. 
\end{proof}

Finally, substituting $M_{T_c}(s)$ from Proposition~\ref{Pro1} and $M_{Y_c}(s)$ from Proposition~\ref{Pro2} into \eqref{MGFofagegeneral} yields the MGF of the AoI, ${M}_{\delta_c}(s)$, as stated in Theorem~\ref{T_source-aware}. Similarly, substituting $M_{T_c}(s)$ and $M_{Y_c}(s)$ into \eqref{MGFpeak} yields the MGF of the PAoI, ${M}_{A_c}(s)$, also given in Theorem~\ref{T_source-aware}.

\begin{figure}
\centering
\includegraphics[width=.95\linewidth,trim = 0mm 0mm 0mm 0mm,clip]{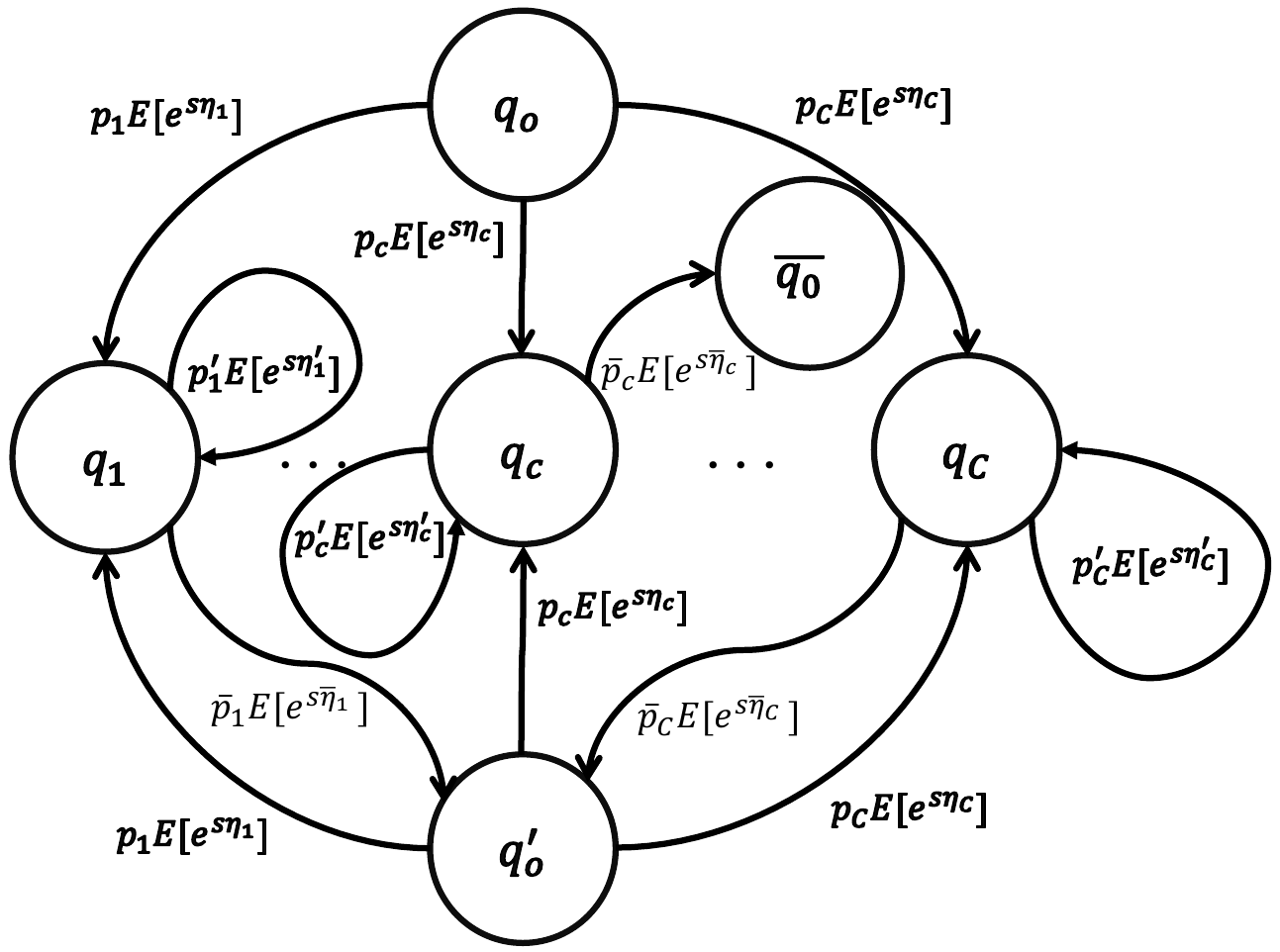}
\caption{The directed graph to calculate the MGF of the interdeparture time under the self-preemptive policy.}  
\label{Detour_c}
\vspace{-7mm}
\end{figure}


\section{Numerical Results}\label{Numerical Results}
In this section, we use Corollary~\ref{Cro_01} to derive the average AoI under the probabilistically preemptive policy in a two-source status update system and compare its performance against the self-preemptive \cite{9869867}, globally-preemptive \cite{8406928}, and non-preemptive \cite{9500775} policies in terms of the sum average AoI.

We assume that the service time $U$ follows a log-normal distribution with PDF
$$
f_U(t) = \frac{1}{t \, \omega \sqrt{2\pi}} \exp\Bigg(-\frac{(\ln t - \alpha)^2}{2 \omega^2}\Bigg), \quad t > 0,
$$
where $\alpha \in (-\infty, \infty)$ and $\omega > 0$. The expected service time is given by $\mathbb{E}[U] = \exp(\alpha + \omega^2/2)$.

In all figures, we set $\lambda = \lambda_1 + \lambda_2 = 8$, $\alpha = -1$, and $\omega = 1$, which gives $\mathbb{E}[U] = 1$.

Fig.~\ref{Versus_P} shows the sum average AoI under different policies as a function of the preemption probability $\theta$, for $\lambda_1 = 2$ (Fig.~\ref{SAoI_2_vsP}), $\lambda_1 = 4$ (Fig.~\ref{SAoI_4_vsP}), and $\lambda_1 = 5$ (Fig.~\ref{SAoI_5_vsP}).

Fig.~\ref{Versus_lambda} shows the sum average AoI difference ratio (in percent) under different policies as a function of the arrival rate of source one, $\lambda_1$, for $\theta = 0.2$ (Fig.~\ref{Ratio_P2_V_lambda}), $\theta = 0.6$ (Fig.~\ref{Ratio_P6_V_lambda}), and $\theta = 0.9$ (Fig.~\ref{Ratio_P9_V_lambda}). The sum average AoI difference ratio (in percent) for each policy is calculated as
\begin{align}
\dfrac{\Delta_1+\Delta_2\!-\!(\Delta^{\text{prob}}_1\!+\!\Delta^{\text{prob}}_2)}{\Delta^{\text{prob}}_1+\Delta^{\text{prob}}_2}\!\times\!100,
\end{align}
where $\Delta^{\text{prob}}_i$ is the average AoI of source $i$ under the probabilistically preemptive policy. From the figures, by properly choosing the preemption probability $\theta$, the sum average AoI can be significantly reduced. For example, when $\lambda_1 = 2$, setting $\theta = 0.28$ results in an improvement of approximately $18\%$ in the sum average AoI compared to the other three policies, as shown in Fig.~\ref{SAoI_2_vsP}.

\begin{figure}
\centering
\subfigure[$\lambda_1=2$]{
\includegraphics[width=0.4\textwidth]{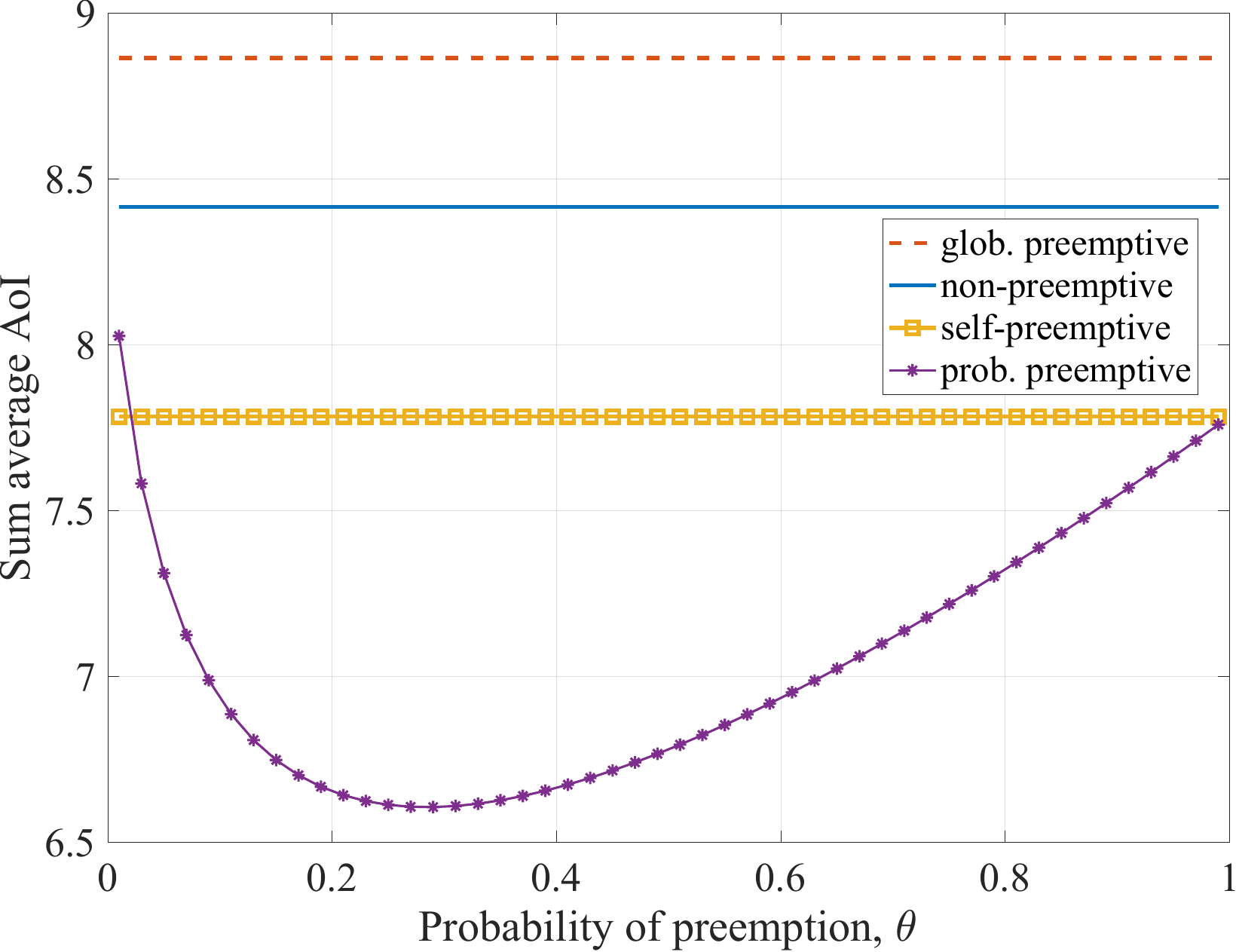}
\label{SAoI_2_vsP}
}
\subfigure[$\lambda_1=4$]
{
\includegraphics[width=0.4\textwidth]{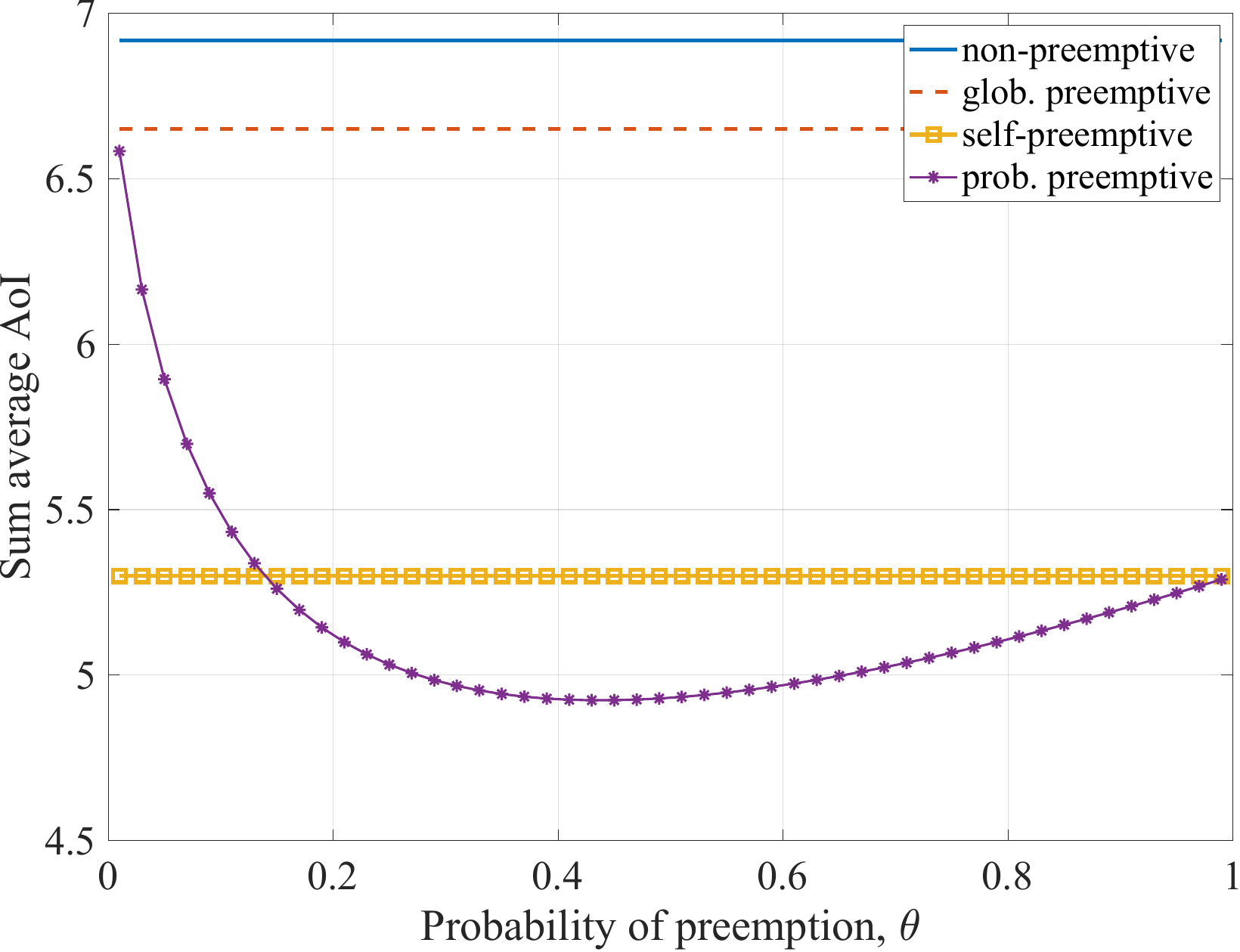}
\label{SAoI_4_vsP}
}
\subfigure[$\lambda_1=5$]{
\includegraphics[width=0.4\textwidth]{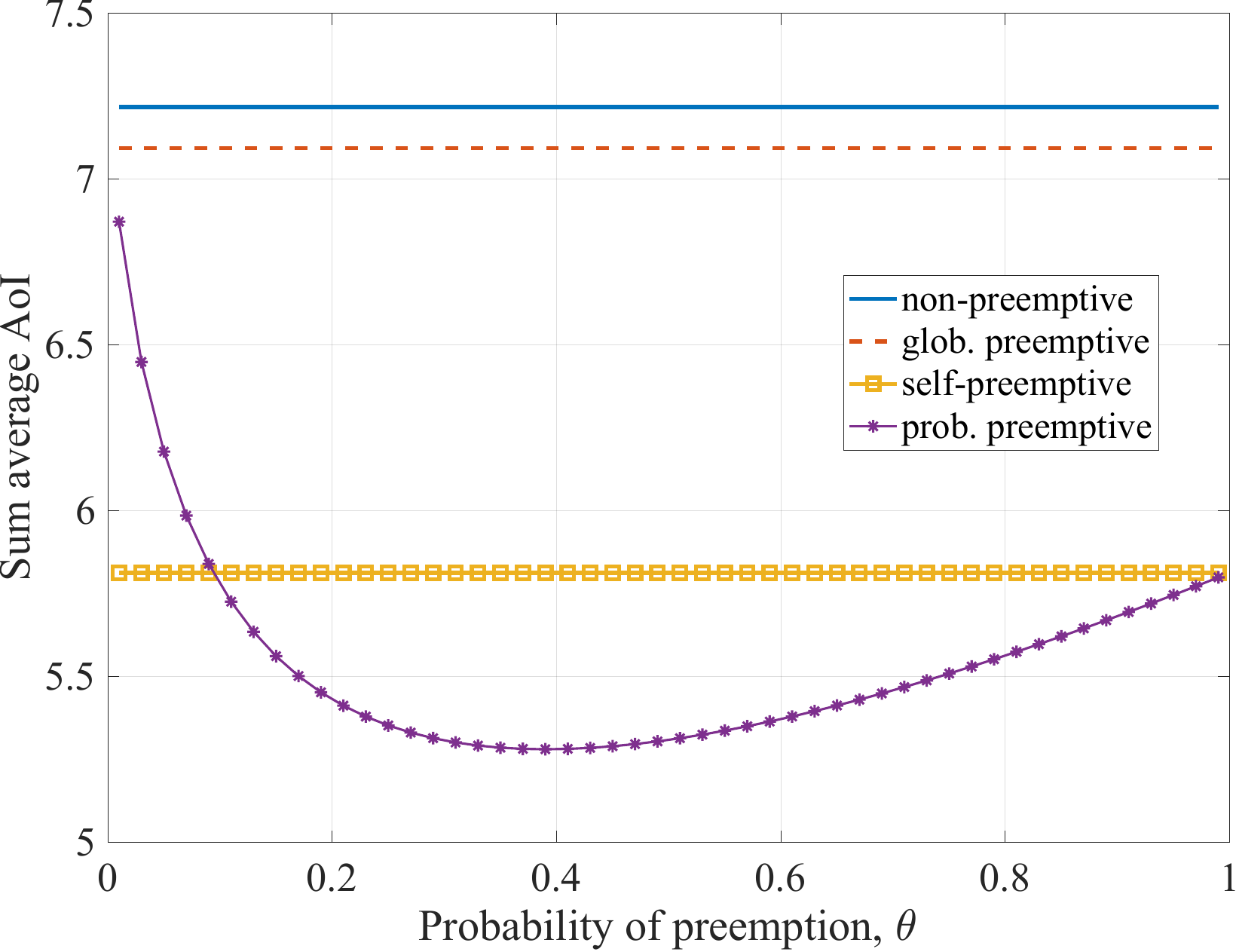}
\label{SAoI_5_vsP}
}
\caption{The sum average AoI of different policies as a function of the probability of preemption $\theta$.}
\label{Versus_P}
\vspace{-6mm}
\end{figure}

\begin{figure}
\centering
\subfigure[$\theta=0.2$]{
\includegraphics[width=0.4\textwidth]{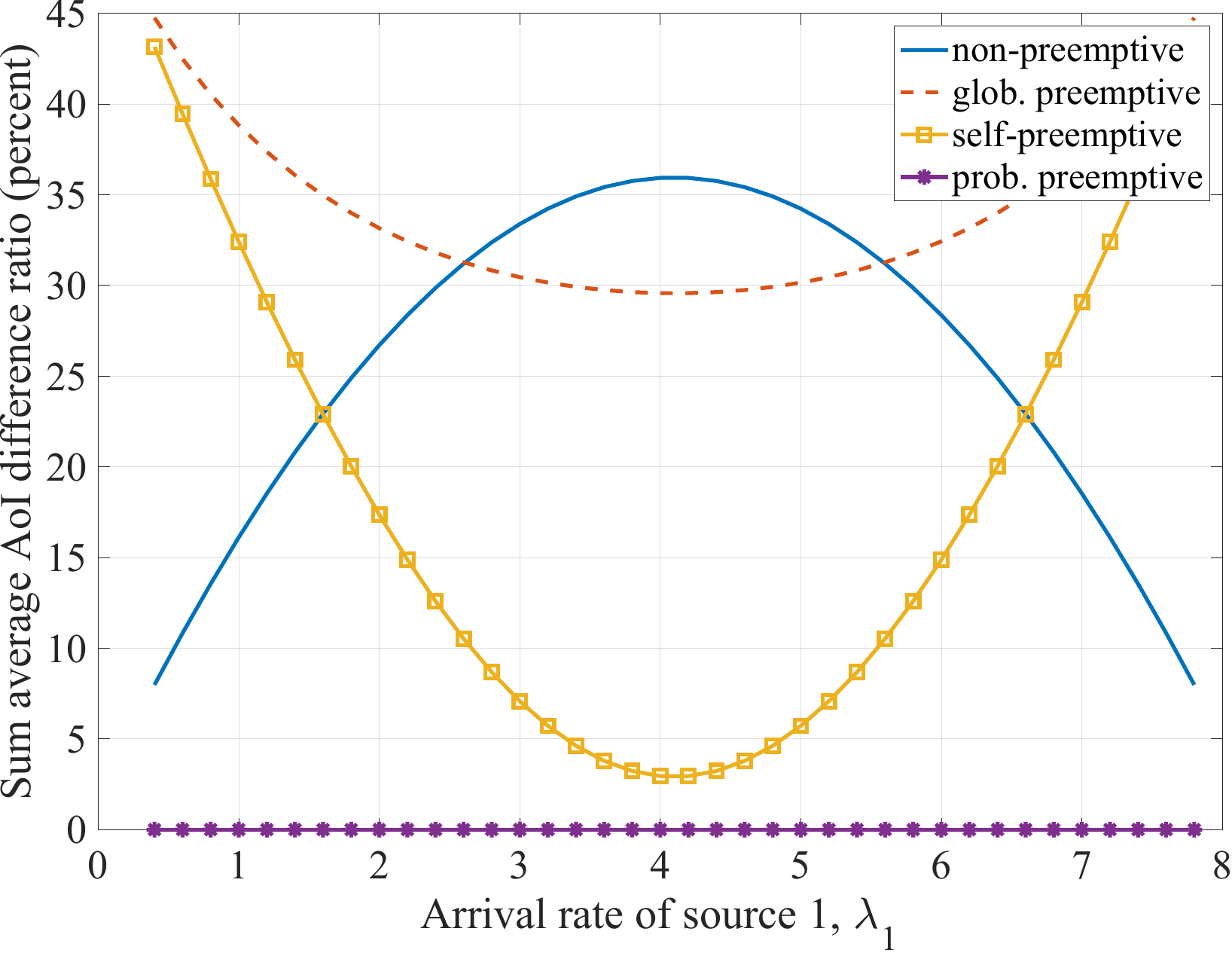}
\label{Ratio_P2_V_lambda}
}
\subfigure[$\theta=0.6$]
{
\includegraphics[width=0.4\textwidth]{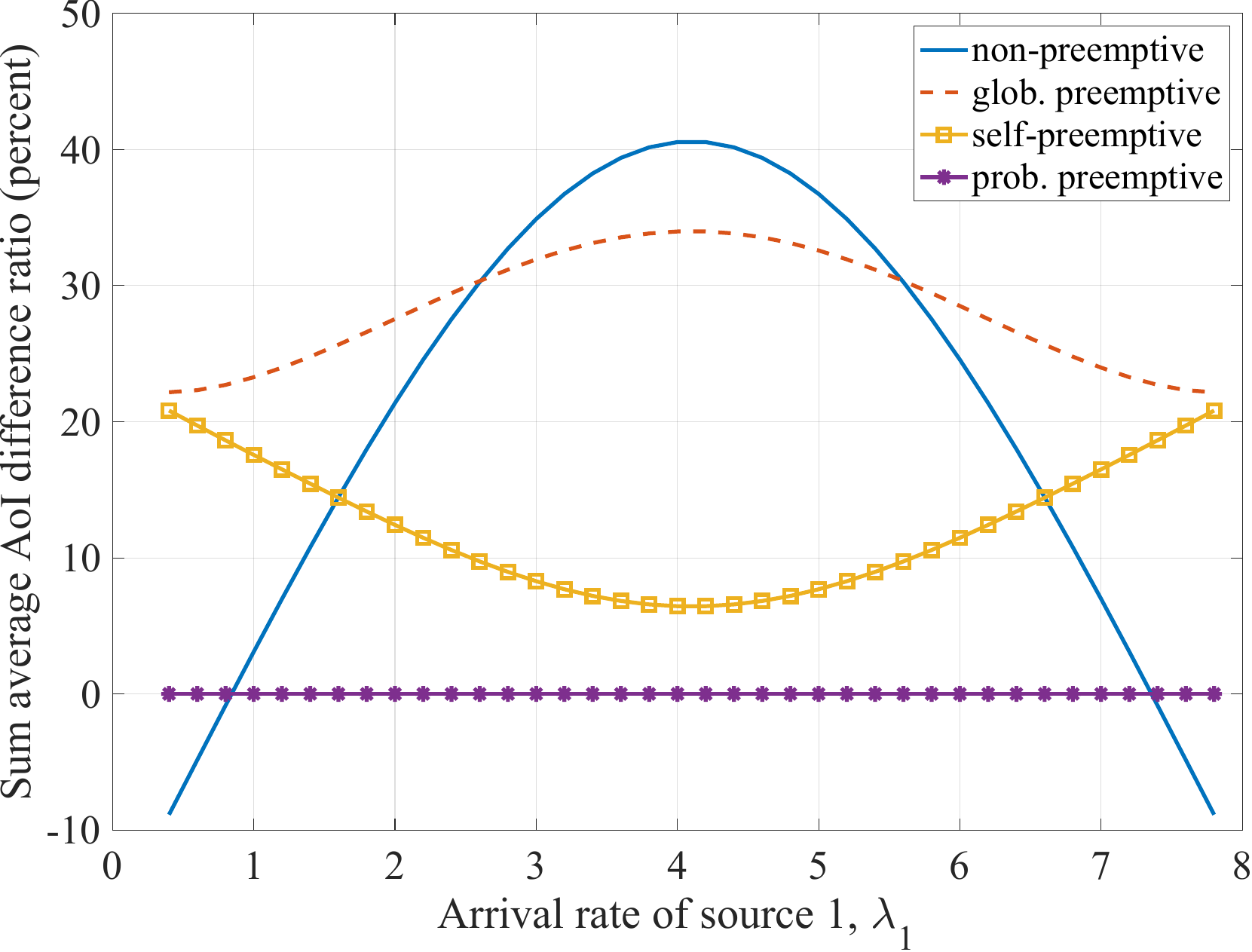}
\label{Ratio_P6_V_lambda}
}
\subfigure[$\theta=0.9$]{
\includegraphics[width=0.4\textwidth]{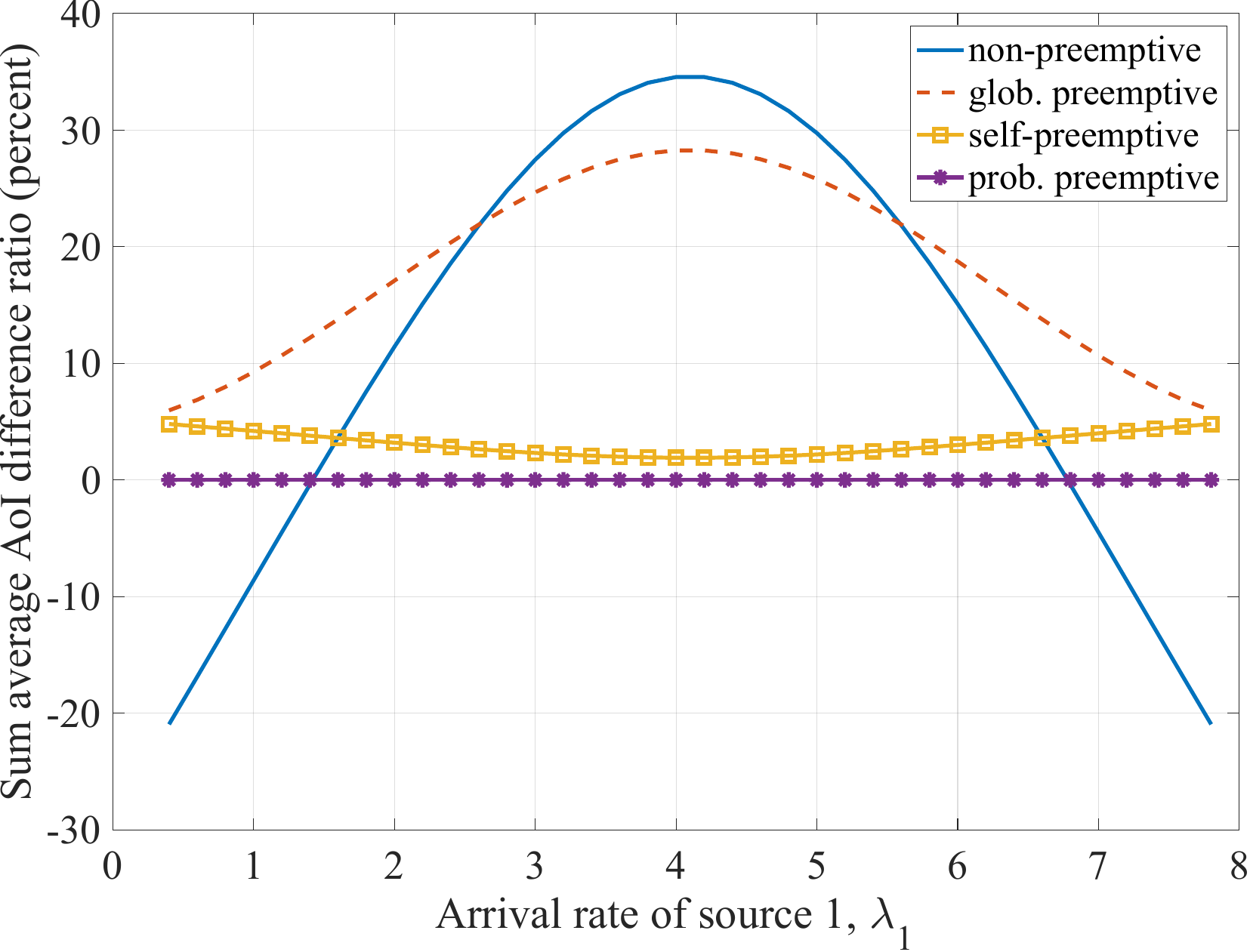}
\label{Ratio_P9_V_lambda}
}
\caption{The sum average AoI difference ratio (percent) for different policies as a function of the arrival rate of source one $\lambda_1$.}
\label{Versus_lambda}
\vspace{-6mm}
\end{figure}


\section{Conclusions}\label{Conclusions}
We studied a multi-source M/G/1/1 queueing system and derived the MGFs of the AoI and PAoI under the probabilistically preemptive policy. Using the AoI MGF, we analyzed the sum average AoI for a log-normal service time distribution. The results demonstrated that the probabilistically preemptive policy can substantially improve system performance compared to existing policies.

\section{Appendix}

\subsection{Proof of Corollary~\ref{Cro_01}}\label{ProofCro_01}
From Lemma~\ref{lemmsmgfage}, we have 
\begin{align}\label{Cro_01_eq_1}
    \bar{Y_{c}}sM_{\delta_{c}}(s)=M_{A_{c}}(s)-M_{T_{c}}(s).
\end{align}
By computing the $(m+1)$th derivative of both sides of \eqref{Cro_01_eq_1}, we obtain
\begin{align}\nonumber
    &\bar{Y_{c}}\left((m+1)M^{(m)}_{\delta_{c}}(s)+sM^{(m+1)}_{\delta_{c}}(s)\right)=\\&\label{Cro_01_eq_2} M^{(m+1)}_{A_{c}}(s)-M^{(m+1)}_{T_{c}}(s).
\end{align}
By taking the limit of both sides of \eqref{Cro_01_eq_2} as $s \to 0$, we obtain
\begin{align}\label{Cro_01_eq_3}
    (m+1)\bar{Y_{c}}\Delta_{c}^{(m)}=A_{c}^{(m+1)}-\E{T_{c}^{m+1}}.
\end{align}
Recall that the PAoI associated with packet $i$ is $A_{c,i} = Y_{c,i} + T_{c,i-1}$. Since there is no waiting buffer in the system, $T_{c,i-1}$ and $Y_{c,i}$ are independent. Therefore, $A_c^{(m)}$ can be expressed as
\begin{align}\label{Cro_01_eq_4}
A_c^{(m)}=\sum_{i=0}^{m}{m \choose i}\E{T_c^{i}}\E{Y_c^{m-i}}.
\end{align}
Finally, substituting \eqref{Cro_01_eq_4} into \eqref{Cro_01_eq_3} completes the proof of Corollary~\ref{Cro_01}.

\subsection{Proof of Proposition~\ref{Pro1}}\label{ProofPro1}
We first derive the PDF of the system time $T_c$, denoted by $f_{T_c}(t)$. Let $D_c$ denote the event that a packet from source $c$ entering service is successfully delivered. According to the packet management policy, the system time $T_c$ of a \textit{delivered packet} is equal to its service time. Thus, the distribution of $T_c$ satisfies
\begin{align}
&\mathrm{Pr}(t\le T_c<t+\epsilon)=\mathrm{Pr}(t\le U\le t+\epsilon\mid D_c)\nn
&\qquad=\frac{\mathrm{Pr}(t<U<t+\epsilon)\mathrm{Pr}(D_c\mid t<U<t+\epsilon)}{\mathrm{Pr}(D_c)}.
\end{align}
 Hence, $f_{T_c}(t)=\lim_{\epsilon\to0}\mathrm{Pr}(t\le T_c\le t+\epsilon)/\epsilon$ is calculated as 
 \begin{align}\label{Sys_time_G1}
  f_{T_c}(t)&= \dfrac{f_U(t)\mathrm{Pr}(D_c\mid U=t)}{\mathrm{Pr}(D_c)}.  
 \end{align}
Because packets of source $c$ arriving at a busy server (serving a packet from source $c$) are admitted into service as Bernoulli trials with probability $\theta$, the arrivals of preempting packets of source $c$, as observed by the busy server, form a thinned Poisson process with arrival rate $\theta \lambda_c$. Let the random variable $R_c$ denote the interarrival time of preempting packets of source $c$, i.e., the time elapsed between any two consecutive preempting packets. The event $D_c$ occurs if and only if this thinned arrival process has zero arrivals during the service period $U$, i.e., $R_c > U$. Thus,
\begin{align}\label{con_pro_eve_eq}
\mathrm{Pr}(D_c\mid U=t)&=\mathrm{Pr}(R_c>U\mid U=t)\nn&\stackrel{(a)}{=}e^{-\theta\lambda_c t},
\end{align}
where $(a)$ follows because $R_c$ follows the exponential distribution with rate $\theta\lambda_c$. This implies 
\begin{align}
\mathrm{Pr}(D_c)&=\int_{0}^{\infty}\mathrm{Pr}(D\mid U=t)f_U(t)\,dt\nn
&=\int_{0}^{\infty}e^{-\theta\lambda_c t}f_U(t)\,dt
=M_{U}(-\theta\lambda_c).
\label{con_pro_eve_eq1}
\end{align}
By substituting \eqref{con_pro_eve_eq} and \eqref{con_pro_eve_eq1} into \eqref{Sys_time_G1}, 
\begin{align}
    f_{T_c}(t)=\dfrac{f_U(t)e^{-\theta\lambda_c t}}{M_{U}(-\theta\lambda_c)},
    \label{f_t1_0}
\end{align}
and the claim follows since $M_{T_c}(s)=\mathbb{E}[e^{sT_c}]$. 

\subsection{Proof of Lemma~\ref{prob01}}\label{Proof_LammaII}
Since $\bar X_{c'}$ is the minimum of independent exponentially distributed random variables ${X_{j},~j\in\mathcal{C}_{-{c'}}}$, it follows an exponential distribution with parameter $\bar\lambda_{c'}=\sum_{j\in\mathcal{C}_{-{c'}}}\lambda_j$. Thus,
\begin{align}\nonumber
p_{c'}&=\mathrm{Pr}(X_{c'}<\bar X_{c'})
\\&\nonumber=\int_{0}^{\infty}\mathrm{Pr}(X_{c'}<\bar X_{c'}\mid \bar X_{c'}=t)f_{\bar X_{c'}}(t)\mathrm{d}t\\&
=\int_{0}^{\infty}(1-e^{-\lambda_{c'}t})\bar\lambda_{c'}e^{-\bar\lambda_{c'}t}\mathrm{d}t=\dfrac{\lambda_{c'}}{\lambda}.
\end{align}
The probability $\bar p_{c'}=\mathrm{Pr}(D_{c'})$ was calculated in \eqref{con_pro_eve_eq1}, and $p'_{c'}=1-\bar p_{c'}$.

\subsection{Proof of Lemma~\ref{lemmfaap}}\label{Proof_LammaII}

The random variable $\bar\eta_{c'}$ corresponds to the system time of source $c'$, whose distribution was derived in \eqref{f_t1_0}.
The PDF of the random variable $\eta_{c'}$ is given as
\begin{align}\nonumber
&f_{\eta_{c'}}(t)=\lim_{\epsilon\rightarrow 0}\dfrac{\mathrm{Pr}(t<\eta_{c'}<t+\epsilon)}{\epsilon}
\\&\nonumber\stackrel{}{=}\lim_{\epsilon\rightarrow 0}\dfrac{\mathrm{Pr}(t<X_{c'}<t+\epsilon\mid X_{c'}<\bar X_{c'})}{\epsilon}\\&\nonumber
=\lim_{\epsilon\rightarrow 0}\dfrac{\mathrm{Pr}(t<X_{c'}\!<\!t+\epsilon)\mathrm{Pr}(X_{c'}\!<\!\bar X_{c'}\mid t\!<\!X_{c'}\!<\!t+\epsilon)}{\epsilon\mathrm{Pr}(X_{c'}<\bar X_{c'})}\\&\label{f_A}
\stackrel{}{=}\dfrac{(1-F_{\bar X_{c'}}(t))f_{X_{c'}}(t)}{\mathrm{Pr}(X_{c'}<\bar X_{c'})}
\stackrel{}{=}\lambda e^{-\lambda t}.
\end{align}

Recall that source $c'$ packets arriving at a busy server (serving a source $c'$ packet) are admitted into service as Bernoulli trials with probability $\theta$. Thus, the arrival time $R_{c'}$ of a preempting packet, as observed by the busy server, is exponentially distributed with parameter $\theta\lambda_{c'}$. The event $\Dbar_{c'}$ occurs when a preemption takes place prior to the completion of the packet in service, i.e., when $U > R_{c'}$. Consequently, the PDF of $\eta'_{c'}$ is given by
\begin{align}\nonumber
&f_{\eta'_{c'}}(t)
\stackrel{}{=}\lim_{\epsilon\rightarrow 0}\dfrac{\mathrm{Pr}(t<R_{c'}<t+\epsilon\mid \Dbar_{c'})}{\epsilon}\\&\nonumber
=\lim_{\epsilon\rightarrow 0}\dfrac{\mathrm{Pr}(t<R_{c'}<t+\epsilon)\mathrm{Pr}(U>R_{c'}\mid t<R_{c'}<t+\epsilon)}{\epsilon\mathrm{Pr}(\Dbar_{c'})}\nn
&=\dfrac{\mathrm{Pr}(U>t)}{\mathrm{Pr}(\Dbar_{c'})}\lim_{\epsilon\rightarrow 0}\dfrac{\mathrm{Pr}(t<R_{c'}<t+\epsilon)}{\epsilon}\nn
&\label{f_t1}
=\dfrac{1-F_U(t)}{1-M_{U}(-\theta\lambda_{c'})}
f_{R_{c'}}(t).
\end{align}

Finally, by substituting the exponential  PDF of $R_{c'}$ (with rate $\theta\lambda_{c'}$), the PDF of $f_{\eta'_{c'}}(t)$ is verified.

\bibliographystyle{IEEEtran}
\bibliography{conf_short,jour_short,Bibliography}

@ARTICLE{10899900,
  author={Akar, Nail and Ulukus, Sennur},
  journal=IEEE_J_COM, 
  title={Age of Information in a Single-Source Generate-at-Will Dual-Server Status Update System}, 
  year={2025},
  volume={73},
  number={9},
  pages={7431-7444},
}

@INPROCEEDINGS{moltafetisit2025,
  author={Mohammad Moltafet and Hamid R. Sadjadpour and Zouheir Rezki and Marian Codreanu and Roy D. Yates},
  booktitle=isit, 
  title={AoI in {M/G/1/1} Queues with Probabilistic Preemption}, 
  year={2025}, 
  address = {Ann Arbor (Michigan), USA},
  month= jun#{22--27},
  pages={1--5}
  }

@ARTICLE{9519697,
  author={Dogan, Ozancan and Akar, Nail},
  journal=IEEE_J_COM, 
  title={The Multi-Source Probabilistically Preemptive {M/PH/1/1} Queue With Packet Errors}, 
  year={2021},
  volume={69},
  number={11},
  pages={7297--7308},
  doi={10.1109/TCOMM.2021.3106347}}

@ARTICLE{9705518,
  author={Abd-Elmagid, Mohamed A. and Dhillon, Harpreet S.},
  journal=IEEE_J_IT, 
  title={Closed-Form Characterization of the {MGF} of {AoI} in Energy Harvesting Status Update Systems}, 
  year={2022},
  volume={68},
  number={6},
  pages={3896-3919},
month=jun,
  }

@ARTICLE{10139823,
  author={Akar, Nail and Gamgam, Ege Orkun},
  journal=IEEE_J_COML, 
  title={Distribution of Age of Information in Status Update Systems With Heterogeneous Information Sources: {An} Absorbing Markov Chain-Based Approach}, 
  year={2023},
  volume={27},
  number={8},
  pages={2024-2028},
month=may,
}

@ARTICLE{10038591,
  author={Fiems, Dieter},
  journal=IEEE_J_COML, 
  title={Age of Information Analysis With Preemptive Packet Management}, 
  year={2023},
  volume={27},
  number={4},
  pages={1105-1109},
month=feb,
 }

@article{inoue2024exact,
      title={Exact Analysis of the Age of Information in the Multi-Source {M/GI/1} Queueing System}, 
      author={Yoshiaki Inoue and Tetsuya Takine},
      year={2024},
      url={https://arxiv.org/abs/2404.05167}, 
}

@ARTICLE{9869867,
  author={Moltafet, Mohammad and Leinonen, Markus and Codreanu, Marian},
  journal=IEEE_J_COM, 
  title={Moment Generating Function of Age of Information in Multisource {M/G/1/1} Queueing Systems}, 
  year={2022},
  volume={70},
  number={10},
  pages={6503-6516},
}

@INPROCEEDINGS{9500775,
  author={Deng, Dapeng and Chen, Zhengchuan and Jia, Yunjian and Liang, Liang and Fang, Shuyang and Wang, Min},
  booktitle=icc, 
  title={Age of Information In A Multiple Stream {M/G/1/1} Non-preemptive Queue}, 
  year={2021},
  volume={},
  number={},
  pages={1--6},
address = {Montreal, QC, Canada},

month=jun  # { 14--23},
  doi={10.1109/ICC42927.2021.9500775}}

@INPROCEEDINGS{8006504,
  author={Najm, Elie and Yates, Roy and Soljanin, Emina},
  booktitle=isit, 
  title={Status updates through {M/G/1/1} queues with {HARQ}}, 
  year={2017},
address = {Aachen, Germany},
  pages={131--135},
month=jun  # { 25--30},
 }

@ARTICLE{9252168,
  author={Moltafet, Mohammad and Leinonen, Markus and Codreanu, Marian},
  journal=IEEE_J_COM, 
  title={Average {AoI} in Multi-Source Systems With Source-Aware Packet Management}, 
  year={2021},
  volume={69},
  number={2},
  pages={1121--1133},
  month=feb,}

@INPROCEEDINGS{9162681,
  author={M. {Moltafet} and M. {Leinonen} and M. {Codreanu}},
  booktitle=infocomw, 
  title={Average Age of Information in a Multi-Source {M/M/1} Queueing Model with {LCFS} Prioritized Packet Management}, 
 address = {Toronto, Canada},
pages={303--308},
month=jul # { 6--9,},
year={2020},}

@ARTICLE{Moltafet2020mgf,
    title={Moment generating function of the {AoI} in a two-source system with packet management},
journal=IEEE_J_WCL,
    author={Mohammad Moltafet and Markus Leinonen and Marian Codreanu},
  year={2021},
  volume={10},
  number={4},
  pages={882--886},
month=apr,
}

@INPROCEEDINGS{9611498,
  author={Moltafet, Mohammad and Leinonen, Markus and Codreanu, Marian},
  booktitle=itw, 
  title={Moment Generating Function of the {AoI} in Multi-Source Systems with Computation-Intensive Status Updates}, 
  address = {Kanazawa, Japan},
pages={1--6},
month=oct # { 17--21,},
year={2021},}

@book{6Bixio2016,
  
title={Principles of Digital Communication: {A} Top-Down Approach},
  
author={Bixio Rimoldi},
  
year={2016},
  
publisher={Cambridge, U.K.: Cambridge University Press} 

}

@INPROCEEDINGS{9048933,
  author={A. {Soysal} and S. {Ulukus}},
  booktitle=asilomar, 
  title={Age of Information in {G/G/1/1} Systems}, 
  year={2019},
address = {Pacific Grove, CA, USA},
month=nov # { 3--6,},
  pages={2022--2027},
}

@INPROCEEDINGS{8406909,
  author={J. P. {Champati} and H. {Al-Zubaidy} and J. {Gross}},
  booktitle=infocomw, 
  title={Statistical guarantee optimization for age of information for the {D/G/1} queue}, 
  year={2018},
  pages={130--135},
address = {Honolulu, HI, USA},
month=apr # { 15--19,},
}

@ARTICLE{8820073,
  author={Y. {Inoue} and H. {Masuyama} and T. {Takine} and T. {Tanaka}},
  journal=IEEE_J_IT, 
  title={A General Formula for the Stationary Distribution of the Age of Information and Its Application to Single-Server Queues}, 
  year={2019},
  volume={65},
  number={12},
  pages={8305--8324},
month=aug,
}

@ARTICLE{8886357,
  author={E. {Najm} and R. {Nasser} and E. {Telatar}},
  journal=IEEE_J_IT, 
  title={Content Based Status Updates}, 
  year={2020},
  volume={66},
  number={6},
  pages={3846--3863},
  month=oct,
}

@ARTICLE{9099557,
  author={M. {Moltafet} and M. {Leinonen} and M. {Codreanu}},
  journal=IEEE_J_COM, 
  title={On the Age of Information in Multi-Source Queueing Models}, 
  year={2020},
  volume={68},
  number={8},
  pages={5003--5017},
  month=may,
}

@ARTICLE{9119460,
  author={N. {Akar} and O. {Dogan} and E. U. {Atay}},
journal=IEEE_J_COM, 
  title={Finding the Exact Distribution of (Peak) Age of Information for Queues of {PH/PH/1/1} and {M/PH/1/2} Type}, 
  year={2020},
  volume={68},
  number={9},
  pages={5661--5672},
month=jun,}

@ARTICLE{7415972,
author={M. Costa and M. Codreanu and A. Ephremides},
journal=IEEE_J_IT,
title={On the Age of Information in Status Update Systems With Packet Management},
year={2016},
volume={62},
number={4},
pages={1897--1910},
month=apr,}

@ARTICLE{8469047,
author={R. D. {Yates} and S. K. {Kaul}},
journal=IEEE_J_IT,
title={The Age of Information: Real-Time Status Updating by Multiple Sources},
year={2019},
volume={65},
number={3},
pages={1807--1827},
month=mar,
}

@INPROCEEDINGS{7541764,
author={E. {Najm} and R. {Nasser}},
booktitle=isit,
title={Age of information: The gamma awakening},
year={2016},
address = {Barcelona, Spain},
pages={2574--2578},
month=jul # { 10--16,},
}

@INPROCEEDINGS{6310931,
author={S. K. Kaul and R. D. Yates and M. Gruteser},
booktitle=ciss,
title={Status updates through queues},
year={2012},
address ={Princeton, NJ, USA},
pages={1--6},
month=mar  # { 21--23,},}

@INPROCEEDINGS{6875100,
author={M. Costa and M. Codreanu and A. Ephremides},
booktitle=isit,
title={Age of Information with Packet Management},
year={2014},
pages={1583--1587},
 address={Honolulu, HI, USA},
month = jun # { 20--23,},
}

@INPROCEEDINGS{8006592,
author={Y. Inoue and H. Masuyama and T. Takine and T. Tanaka},
booktitle=isit,
title={The stationary distribution of the age of information in {FCFS} single-server queues},
year={2017},
address = {Aachen, Germany},
pages={571--575},
month=jun # { 25--30,},}

@INPROCEEDINGS{8406928,
author={E. Najm and E. Telatar},
booktitle=infocom,
title={Status updates in a multi-stream {M/G/1/1} preemptive queue},
year={2018},
address = {Honolulu, HI, USA},
pages={124--129},
month=apr # { 15--19,},}

@INPROCEEDINGS{8406966,
author={R. D. {Yates}},
booktitle=infocom,
title={Age of information in a network of preemptive servers},
year={2018},
address = {Honolulu, HI, USA},
pages={118--123},
month=apr # { 15--19},
}

@INPROCEEDINGS{6195689,
author={S. Kaul and R. Yates and M. Gruteser},
booktitle=infocom,
title={Real-time status: How often should one update?},
year={2012},
address = {Orlando, FL, USA},
pages={2731--2735},
month=mar # { 25--30,},
}

@INPROCEEDINGS{6284003,
author={R. D. Yates and S. Kaul},
booktitle=isit,
title={Real-time status updating: Multiple sources},
address = {Cambridge, MA, USA},
pages={2666--2670},
month=jul # { 1--6,},
year={2012},
}

@INPROCEEDINGS{9013935,
  author={A. {Javani} and M. {Zorgui} and Z. {Wang}},
  booktitle=globecom, 
  title={Age of Information in Multiple Sensing}, 
  year={2019},
address = {Waikoloa, HI, USA},
pages={1--6},
month=dec # { 9--13,},
}

@INPROCEEDINGS{9048914,
  author={S. {Farazi} and A. G. {Klein} and D. {Richard Brown}},
  booktitle=asilomar, 
  title={Average Age of Information in Multi-Source Self-Preemptive Status Update Systems with Packet Delivery Errors}, 
  year={2019},
address = {Pacific Grove, CA, USA},
  pages={396--400},
month= nov # { 3--6,},
}

@INPROCEEDINGS{8437907,
author={R. D. {Yates}},
booktitle=isit,
title={Status Updates through Networks of Parallel Servers},
year={2018},
pages={2281--2285},
address = {Vail, CO, USA},
month=jun  # { 17--22,}
}

@INPROCEEDINGS{8437591,
author={S. K. {Kaul} and R. D. {Yates}},
booktitle=isit,
title={Age of Information: Updates with Priority},
year={2018},
pages={2644--2648},
address = {Vail, CO, USA},
month =  jun # { 17--22,},
}

@STRING{IEEE_J_COML       = "{IEEE} Commun. Lett."}

@STRING{IEEE_J_COM        = "{IEEE} Trans. Commun."}

@STRING{IEEE_J_IT         = "{IEEE} Trans. Inf. Theory"}

@string{ globecom = {Proceedings of the IEEE Global Telecommunication Conference}}

@string{ icc = {Proceedings of the IEEE International Conference on Communications}}

@string{ isit = {Proceedings of the IEEE International Symposium on Information Theory}}

@string{ itw = {Proceedings of the IEEE Information Theory Workshop}}

@string{ infocom = {Proceedings of the IEEE International Conference on Computer Communications}}

@string{ asilomar = {Proceedings of the Annual Asilomar Conference on Signals, Systems and Computers}}

@string{ ciss = {Proceedings of the Conference on Information Sciences and Systems}}

@string{ infocom= {Proceedings of the International Conference on Computer Communications (INFOCOM)}}

@string{ infocomw = {Proceedings of the International Conference on Computer Communications (INFOCOM) Workshop}}

@string{ globecom = {Proc. IEEE Global Telecommun. Conf.}}

@string{ icc = {Proc. IEEE Int. Conf. Commun.}}

@string{ isit = {Proc. IEEE Int. Symp. Inform. Theory}}

@string{ itw = {Proc. IEEE Inform. Theory Workshop}}

@string{ infocom = {Proc. IEEE Int. Conf. on Comp. Commun.}}

@string{ asilomar = {Proc. Annual Asilomar Conf. Signals, Syst., Comp.}}

@string{ infocom = {Proc. IEEE Int. Conf. on Computer. Commun.  (INFOCOM)}}

@string{ infocomw = {Proc. IEEE Int. Conf. on Computer. Commun.  (INFOCOM) Workshop}}

@string{ ciss = {Proc. Conf. Inform. Sciences Syst. (CISS)}}

@STRING{IEEE_J_COML       = "{IEEE} Communications Letters"}

@STRING{IEEE_J_COM        = "{IEEE} Transactions on Communications"}

@STRING{IEEE_J_IT         = "{IEEE} Transactions on Information Theory"}

@STRING{IEEE_J_WCL       = "{IEEE} Wireless Commun. Lett."}

@STRING{IEEE_J_IT         = "{IEEE} Trans. Inform. Theory"}
\end{document}